\documentclass[11pt]{article}

\usepackage[T1]{fontenc}
\usepackage[utf8]{inputenc}
\usepackage{lmodern}
\usepackage{microtype}
\usepackage[a4paper,margin=1in]{geometry}
\usepackage{amsmath,amssymb,amsthm,mathtools}
\usepackage{xcolor}
\usepackage{hyperref}
\usepackage{enumitem}

\hypersetup{
  colorlinks=true,
  linkcolor=blue!50!black,
  citecolor=blue!50!black,
  urlcolor=blue!50!black,
  pdftitle={Three Fixed-Dimension Satisfiability Semantics for Quantum Logic},
  pdfauthor={Anonymous}
}

\theoremstyle{plain}
\newtheorem{theorem}{Theorem}[section]

\newtheorem{lemma}[theorem]{Lemma}
\newtheorem{corollary}[theorem]{Corollary}

\theoremstyle{definition}
\newtheorem{definition}[theorem]{Definition}
\newtheorem{remark}[theorem]{Remark}
\newtheorem{problem}[theorem]{Problem}

\newcommand{\STD}{\mathrm{STD}}
\newcommand{\COM}{\mathrm{COM}}
\newcommand{\PBA}{\mathrm{PBA}}
\newcommand{\Sat}{\mathrm{Sat}}
\newcommand{\Proj}{\mathrm{Proj}}
\newcommand{\Ran}{\mathrm{Ran}}
\newcommand{\Var}{\mathrm{Var}}
\newcommand{\SEPone}{\mathrm{SEP}\text{-}1}
\newcommand{\sem}[1]{\left[\!\left[#1\right]\!\right]}
\newcommand{\field}{\mathbb{F}}

\title{Three Fixed-Dimension Satisfiability Semantics for Quantum Logic\\[0.35em]
\large Implications and an Explicit Separator}
\author{Joaquim Reizi Higuchi\\
Program in Nature and Environment, Graduate School of Arts and Sciences\\
The Open University of Japan, Chiba, Japan\\
\texttt{1218237360@campus.ouj.ac.jp}}
\date{March 6, 2026}

\begin{document}
\maketitle

\begin{abstract}
We compare three satisfiability notions for propositional formulas in the language
$\{\neg,\wedge,\vee\}$ over a fixed finite-dimensional Hilbert space $H=\field^d$
with $\field\in\{\mathbb{R},\mathbb{C}\}$.
The first is the standard Hilbert-lattice semantics on the subspace lattice $L(H)$,
where meet and join are total operations.
The second is a global commuting-projector semantics, where all atoms occurring in the formula are interpreted by a single pairwise-commuting projector family.
The third is a local partial-Boolean semantics, where binary connectives are defined only on commeasurable pairs and definedness is checked nodewise along the parse tree.

We prove, for every fixed $d\geq 1$,
\[
\Sat_{\COM}^d(\varphi) \Longrightarrow \Sat_{\PBA}^d(\varphi) \Longrightarrow \Sat_{\STD}^d(\varphi)
\]
for every formula $\varphi$.
We then exhibit the explicit formula
\[
\SEPone := (p\wedge(q\vee r))\wedge \neg\bigl((p\wedge q)\vee(p\wedge r)\bigr)
\]
which is satisfiable in the standard semantics for every $d\geq 2$, but unsatisfiable under both the global commuting and the partial-Boolean semantics.
Consequently, for every $d\geq 2$, the satisfiability classes satisfy
\[
\mathsf{SAT}_{\COM}^d \subseteq \mathsf{SAT}_{\PBA}^d \subsetneq \mathsf{SAT}_{\STD}^d,
\qquad
\mathsf{SAT}_{\COM}^d \subsetneq \mathsf{SAT}_{\STD}^d,
\]
while the exact relation between $\mathsf{SAT}_{\COM}^d$ and $\mathsf{SAT}_{\PBA}^d$ remains open.
The point of the paper is semantic comparison, not a new feasibility reduction or a generic translation theorem.
\end{abstract}

\section{Introduction}

There are at least three natural ways to interpret the propositional connectives
$\neg,\wedge,\vee$ in finite-dimensional quantum logic.
The first is the standard Hilbert-lattice semantics on the lattice of subspaces of a Hilbert space.
The second restricts attention to a single commuting family of projectors, so that the connectives are interpreted inside a Boolean block.
The third uses the partial Boolean algebra of projectors and allows binary connectives only on commeasurable pairs.
These three semantics are closely related, but they are not the same, and conflating them obscures what is actually being proved.

This distinction matters because several adjacent lines of prior work already occupy nearby territory.
Herrmann showed that, for finite-height complemented and orthocomplemented modular lattices, satisfiability and the complement of equivalence are polynomial-time equivalent to algebraic feasibility in the subspace-lattice case of dimension at least three~\cite{Herrmann2022}.
Thus a generic claim of novelty based only on fixed-dimension feasibility reduction for the standard Hilbert-lattice semantics would be misplaced.
Herrmann and Ziegler developed substantial definability and translation machinery for finite-dimensional subspace lattices with involution~\cite{HerrmannZiegler2018,HerrmannZiegler2019}, so generic field-to-lattice translation claims are likewise not new.
On the other hand, Dawar and Shah recently proved strong complexity theorems for satisfiability in Kochen-Specker partial Boolean algebras, including existential-theory-of-the-reals completeness in fixed dimension for the relevant real and complex cases~\cite{DawarShah2026}.
That result concerns the partial-Boolean semantics rather than the standard Hilbert-lattice semantics, so it should be treated as a comparison target, not as a substitute.

The present paper makes a narrower, cleaner contribution.
We formalize the three semantics precisely, prove the straightforward implication chain
\[
\COM \Rightarrow \PBA \Rightarrow \STD,
\]
and then separate the standard semantics from the other two by one explicit formula.
The remaining unresolved comparison is whether the local partial-Boolean semantics is genuinely more permissive than the global commuting semantics.

Our main results are the following.

\begin{enumerate}[label=(\arabic*),leftmargin=2.1em]
  \item For every $d\geq 1$ and every formula $\varphi$,
  \[
  \Sat_{\COM}^d(\varphi) \Longrightarrow \Sat_{\PBA}^d(\varphi) \Longrightarrow \Sat_{\STD}^d(\varphi).
  \]
  \item For every $d\geq 2$, the formula
  \[
  \SEPone := (p\wedge(q\vee r))\wedge \neg\bigl((p\wedge q)\vee(p\wedge r)\bigr)
  \]
  satisfies
  \[
  \Sat_{\STD}^d(\SEPone),
  \qquad
  \neg\Sat_{\COM}^d(\SEPone),
  \qquad
  \neg\Sat_{\PBA}^d(\SEPone).
  \]
  \item Consequently,
  \[
  \mathsf{SAT}_{\PBA}^d \subsetneq \mathsf{SAT}_{\STD}^d
  \qquad\text{and}\qquad
  \mathsf{SAT}_{\COM}^d \subsetneq \mathsf{SAT}_{\STD}^d
  \]
  for every $d\geq 2$, while the exact relation between $\mathsf{SAT}_{\COM}^d$ and $\mathsf{SAT}_{\PBA}^d$ is left open.
\end{enumerate}

\section{Syntax and the three semantics}

Fix a field $\field\in\{\mathbb{R},\mathbb{C}\}$ and an integer $d\geq 1$.
Write $H:=\field^d$ with its standard inner product.
Let $L(H)$ be the set of linear subspaces of $H$, and let $\Proj(H)$ be the set of orthogonal projectors on $H$.

\begin{definition}[Formulas]
Fix a countable set of atoms
\[
\Var=\{p_0,p_1,p_2,\dots\}.
\]
The set of formulas is generated by the grammar
\[
\varphi ::= p_i \mid \neg\varphi \mid (\varphi\wedge\psi) \mid (\varphi\vee\psi).
\]
No constants and no additional connectives are assumed.
\end{definition}

\begin{definition}[Standard Hilbert-lattice semantics]
A \emph{standard valuation} is a map $v:\Var\to L(H)$.
For such a valuation, define $\sem{\varphi}^{\STD}_v\in L(H)$ recursively by
\[
\sem{p_i}^{\STD}_v := v(p_i),
\qquad
\sem{\neg\varphi}^{\STD}_v := \bigl(\sem{\varphi}^{\STD}_v\bigr)^{\perp},
\]
\[
\sem{\varphi\wedge\psi}^{\STD}_v := \sem{\varphi}^{\STD}_v\cap \sem{\psi}^{\STD}_v,
\qquad
\sem{\varphi\vee\psi}^{\STD}_v := \sem{\varphi}^{\STD}_v + \sem{\psi}^{\STD}_v.
\]
We write $\Sat_{\STD}^d(\varphi)$ if there exists a standard valuation $v$ such that
\[
\sem{\varphi}^{\STD}_v \neq \{0\}.
\]
\end{definition}

\begin{remark}
Since $H$ is finite-dimensional, every linear subspace is automatically closed, and the sum of two subspaces is again a subspace.
Thus the join clause above is well-defined inside $L(H)$.
\end{remark}

\begin{definition}[Global commuting-projector semantics]
A \emph{projector valuation} is a map $v:\Var\to \Proj(H)$.
Given a formula $\varphi$, such a valuation is \emph{COM-admissible for $\varphi$} if for every pair of atoms $p_i,p_j$ occurring in $\varphi$ one has
\[
 v(p_i)v(p_j)=v(p_j)v(p_i).
\]
If $v$ is COM-admissible for $\varphi$, define $\sem{\theta}^{\COM}_v\in\Proj(H)$ for every subformula $\theta$ of $\varphi$ by
\[
\sem{p_i}^{\COM}_v := v(p_i),
\qquad
\sem{\neg\theta}^{\COM}_v := I-\sem{\theta}^{\COM}_v,
\]
\[
\sem{\alpha\wedge\beta}^{\COM}_v := \sem{\alpha}^{\COM}_v\,\sem{\beta}^{\COM}_v,
\qquad
\sem{\alpha\vee\beta}^{\COM}_v := \sem{\alpha}^{\COM}_v + \sem{\beta}^{\COM}_v - \sem{\alpha}^{\COM}_v\sem{\beta}^{\COM}_v.
\]
We write $\Sat_{\COM}^d(\varphi)$ if there exists a COM-admissible valuation $v$ for $\varphi$ such that
\[
\sem{\varphi}^{\COM}_v \neq 0.
\]
\end{definition}

\begin{definition}[Local partial-Boolean semantics]
For projectors $P,Q\in\Proj(H)$ write $P\# Q$ if $PQ=QP$.
A projector valuation is again a map $v:\Var\to\Proj(H)$.
For a formula $\varphi$ we define, recursively, what it means for $\sem{\varphi}^{\PBA}_v$ to be defined, and when defined we set its value as follows:
\begin{align*}
\sem{p_i}^{\PBA}_v &:= v(p_i),\\
\sem{\neg\theta}^{\PBA}_v &:= I-\sem{\theta}^{\PBA}_v \qquad \text{if } \sem{\theta}^{\PBA}_v \text{ is defined},\\
\sem{\alpha\wedge\beta}^{\PBA}_v &:= \sem{\alpha}^{\PBA}_v\,\sem{\beta}^{\PBA}_v \qquad \text{if } \sem{\alpha}^{\PBA}_v,\sem{\beta}^{\PBA}_v \text{ are defined and } \sem{\alpha}^{\PBA}_v \# \sem{\beta}^{\PBA}_v,\\
\sem{\alpha\vee\beta}^{\PBA}_v &:= \sem{\alpha}^{\PBA}_v + \sem{\beta}^{\PBA}_v - \sem{\alpha}^{\PBA}_v\sem{\beta}^{\PBA}_v \\
&\hspace{3.3cm}\text{if } \sem{\alpha}^{\PBA}_v,\sem{\beta}^{\PBA}_v \text{ are defined and } \sem{\alpha}^{\PBA}_v \# \sem{\beta}^{\PBA}_v.
\end{align*}
We write $\Sat_{\PBA}^d(\varphi)$ if there exists a projector valuation $v$ such that $\sem{\varphi}^{\PBA}_v$ is defined and
\[
\sem{\varphi}^{\PBA}_v\neq 0.
\]
\end{definition}

\begin{remark}
The distinction between \(\COM\) and \(\PBA\) is structural.
In \(\COM\), commutativity is required globally at the atom level for the whole formula.
In \(\PBA\), commutativity is checked only where a binary connective is actually evaluated.
The standard semantics imposes no commutativity restriction at all.
\end{remark}

\section{The easy implications}

We first record the range identities needed to pass from projector semantics to subspace semantics.

\begin{lemma}[Range identities for commuting projections]\label{lem:range-identities}
Let $P,Q\in\Proj(H)$ satisfy $PQ=QP$.
Then
\[
\Ran(PQ)=\Ran(P)\cap \Ran(Q),
\qquad
\Ran(I-P)=\Ran(P)^{\perp},
\]
and
\[
\Ran(P+Q-PQ)=\Ran(P)+\Ran(Q).
\]
\end{lemma}

\begin{proof}
For the intersection identity, let $x\in\Ran(PQ)$.
Then $x=PQy$ for some $y\in H$.
Since $P$ and $Q$ commute,
\[
Px=PPQy=PQy=x,
\qquad
Qx=QPQy=PQy=x.
\]
Thus $x\in\Ran(P)\cap\Ran(Q)$.
Conversely, if $x\in\Ran(P)\cap\Ran(Q)$, then $Px=x$ and $Qx=x$, hence
\[
PQx=x,
\]
so $x\in\Ran(PQ)$.
Therefore $\Ran(PQ)=\Ran(P)\cap\Ran(Q)$.

The complement identity is immediate: $I-P$ is the orthogonal projector onto $\Ran(P)^{\perp}$.

For the join identity, set
\[
J:=P+Q-PQ = Q + P(I-Q).
\]
Because $P$ commutes with $Q$, the operator $P(I-Q)$ is an orthogonal projector.
Moreover,
\[
Q\,P(I-Q)=PQ-PQ^2=PQ-PQ=0,
\]
and similarly $P(I-Q)Q=0$.
Hence $Q$ and $P(I-Q)$ are orthogonal projections, so $J$ is the orthogonal projector onto
\[
\Ran(Q)\oplus \Ran(P(I-Q)).
\]
It remains to identify $\Ran(P(I-Q))$.
If $x=P(I-Q)y$, then $Px=x$ and
\[
Qx=QP(I-Q)y=PQ(I-Q)y=0,
\]
so $x\in\Ran(P)\cap\ker(Q)$.
Conversely, if $x\in\Ran(P)\cap\ker(Q)$, then $Px=x$ and $Qx=0$, hence
\[
P(I-Q)x=Px-PQx=x.
\]
Thus
\[
\Ran(P(I-Q))=\Ran(P)\cap\ker(Q).
\]
Therefore
\[
\Ran(J)=\Ran(Q)\oplus\bigl(\Ran(P)\cap\ker(Q)\bigr)=\Ran(P)+\Ran(Q),
\]
as required.
\end{proof}

\begin{lemma}[COM evaluations are also PBA evaluations]\label{lem:com-to-pba-eval}
Let $\varphi$ be a formula and let $v$ be COM-admissible for $\varphi$.
Then for every subformula $\theta$ of $\varphi$:
\begin{enumerate}[label=(\roman*),leftmargin=2em]
  \item $\sem{\theta}^{\COM}_v$ is an orthogonal projector;
  \item $\sem{\theta}^{\COM}_v$ belongs to the commutative algebra generated by the atom values $v(p_i)$ occurring in $\varphi$;
  \item $\sem{\theta}^{\PBA}_v$ is defined and
  \[
  \sem{\theta}^{\PBA}_v = \sem{\theta}^{\COM}_v.
  \]
\end{enumerate}
\end{lemma}

\begin{proof}
We argue by induction on the structure of $\theta$.
For an atom $p_i$ there is nothing to prove.

Assume the claim holds for $\theta$.
Then $\sem{\theta}^{\COM}_v$ is a projector in the commutative algebra generated by the atom values, so
\[
I-\sem{\theta}^{\COM}_v
\]
is again a projector in the same algebra.
This proves the claim for $\neg\theta$.

Now let $\theta=(\alpha\wedge\beta)$ or $\theta=(\alpha\vee\beta)$, and assume the claim for $\alpha$ and $\beta$.
By induction, $\sem{\alpha}^{\COM}_v$ and $\sem{\beta}^{\COM}_v$ are projectors in the same commutative algebra, hence they commute.
Therefore their product is a projector, and so is
\[
\sem{\alpha}^{\COM}_v + \sem{\beta}^{\COM}_v - \sem{\alpha}^{\COM}_v\sem{\beta}^{\COM}_v.
\]
Both operators belong to the same commutative algebra.
This proves (i) and (ii) for the binary cases.

Because the child values commute, the corresponding PBA node is defined.
The PBA and COM recursion clauses are literally the same on commuting projectors, so the values coincide.
This proves (iii).
\end{proof}

\begin{theorem}[\(\COM\Rightarrow\PBA\)]\label{thm:com-implies-pba}
For every $d\geq 1$ and every formula $\varphi$,
\[
\Sat_{\COM}^d(\varphi) \Longrightarrow \Sat_{\PBA}^d(\varphi).
\]
\end{theorem}

\begin{proof}
Assume $\Sat_{\COM}^d(\varphi)$.
Then there exists a COM-admissible valuation $v$ such that $\sem{\varphi}^{\COM}_v\neq 0$.
By Lemma~\ref{lem:com-to-pba-eval}, the PBA evaluation is defined and satisfies
\[
\sem{\varphi}^{\PBA}_v = \sem{\varphi}^{\COM}_v \neq 0.
\]
Hence $\Sat_{\PBA}^d(\varphi)$.
\end{proof}

\begin{theorem}[\(\PBA\Rightarrow\STD\)]\label{thm:pba-implies-std}
For every $d\geq 1$ and every formula $\varphi$,
\[
\Sat_{\PBA}^d(\varphi) \Longrightarrow \Sat_{\STD}^d(\varphi).
\]
\end{theorem}

\begin{proof}
Assume $\Sat_{\PBA}^d(\varphi)$.
Then there exists a projector valuation $v$ such that $\sem{\varphi}^{\PBA}_v$ is defined and nonzero.
Define a standard valuation $u:\Var\to L(H)$ by
\[
 u(p_i):=\Ran(v(p_i)).
\]
We claim that for every subformula $\theta$ of $\varphi$,
\[
\sem{\theta}^{\STD}_u = \Ran\bigl(\sem{\theta}^{\PBA}_v\bigr).
\]
We prove this by induction on $\theta$.

For an atom $p_i$, this is the definition of $u$.
For negation, use the induction hypothesis and the identity
\[
\Ran(I-P)=\Ran(P)^{\perp}
\]
from Lemma~\ref{lem:range-identities}.

Let $\theta=(\alpha\wedge\beta)$.
Because the PBA value of $\theta$ is defined, the values
$\sem{\alpha}^{\PBA}_v$ and $\sem{\beta}^{\PBA}_v$ commute.
Hence Lemma~\ref{lem:range-identities} yields
\[
\Ran\bigl(\sem{\theta}^{\PBA}_v\bigr)
=
\Ran\bigl(\sem{\alpha}^{\PBA}_v\sem{\beta}^{\PBA}_v\bigr)
=
\Ran\bigl(\sem{\alpha}^{\PBA}_v\bigr) \cap \Ran\bigl(\sem{\beta}^{\PBA}_v\bigr).
\]
By the induction hypothesis this equals
\[
\sem{\alpha}^{\STD}_u \cap \sem{\beta}^{\STD}_u = \sem{\theta}^{\STD}_u.
\]

Let $\theta=(\alpha\vee\beta)$.
Again the PBA value is defined only on commuting child values, so by Lemma~\ref{lem:range-identities},
\[
\Ran\bigl(\sem{\theta}^{\PBA}_v\bigr)
=
\Ran\Bigl(\sem{\alpha}^{\PBA}_v + \sem{\beta}^{\PBA}_v - \sem{\alpha}^{\PBA}_v\sem{\beta}^{\PBA}_v\Bigr)
=
\Ran\bigl(\sem{\alpha}^{\PBA}_v\bigr) + \Ran\bigl(\sem{\beta}^{\PBA}_v\bigr).
\]
The induction hypothesis gives
\[
\Ran\bigl(\sem{\theta}^{\PBA}_v\bigr)
=
\sem{\alpha}^{\STD}_u + \sem{\beta}^{\STD}_u
=
\sem{\theta}^{\STD}_u.
\]
This proves the claim.

Applying the claim to $\theta=\varphi$ yields
\[
\sem{\varphi}^{\STD}_u = \Ran\bigl(\sem{\varphi}^{\PBA}_v\bigr).
\]
Since $\sem{\varphi}^{\PBA}_v\neq 0$, its range is nonzero, so
\[
\sem{\varphi}^{\STD}_u \neq \{0\}.
\]
Therefore $\Sat_{\STD}^d(\varphi)$.
\end{proof}

\begin{corollary}[Implication chain]\label{cor:implication-chain}
For every $d\geq 1$ and every formula $\varphi$,
\[
\Sat_{\COM}^d(\varphi) \Longrightarrow \Sat_{\PBA}^d(\varphi) \Longrightarrow \Sat_{\STD}^d(\varphi).
\]
In particular,
\[
\Sat_{\COM}^d(\varphi) \Longrightarrow \Sat_{\STD}^d(\varphi).
\]
\end{corollary}

\begin{proof}
Combine Theorems~\ref{thm:com-implies-pba} and~\ref{thm:pba-implies-std}.
\end{proof}

\section{The separator \texorpdfstring{$\SEPone$}{SEP-1}}

\begin{definition}[The separator formula]
Define
\[
\SEPone := (p\wedge(q\vee r))\wedge \neg\bigl((p\wedge q)\vee(p\wedge r)\bigr).
\]
\end{definition}

The point of $\SEPone$ is simple.
In the standard Hilbert-lattice semantics, the failure of distributivity can make
\[
 p\wedge(q\vee r)
\quad\text{strictly larger than}\quad
 (p\wedge q)\vee(p\wedge r).
\]
In contrast, once the valuation lives inside one commuting Boolean block, the two sides coincide.

\begin{theorem}[\(\SEPone\) is COM-unsatisfiable]\label{thm:sep1-com-unsat}
For every $d\geq 1$,
\[
\neg\Sat_{\COM}^d(\SEPone).
\]
\end{theorem}

\begin{proof}
Let $v$ be COM-admissible for $\SEPone$, and set
\[
A:=v(p),\qquad B:=v(q),\qquad C:=v(r).
\]
Then $A,B,C$ are pairwise commuting orthogonal projectors.
Compute the two inner terms of $\SEPone$ under the COM semantics.
First,
\[
\sem{p\wedge(q\vee r)}^{\COM}_v
=
A(B+C-BC)
=
AB+AC-ABC.
\]
Second,
\[
\sem{(p\wedge q)\vee(p\wedge r)}^{\COM}_v
=
AB+AC-(AB)(AC).
\]
Since $A,B,C$ commute and $A^2=A$,
\[
(AB)(AC)=A^2BC=ABC.
\]
Hence
\[
\sem{p\wedge(q\vee r)}^{\COM}_v
=
\sem{(p\wedge q)\vee(p\wedge r)}^{\COM}_v.
\]
Denote this common projector by $X$.
Then
\[
\sem{\SEPone}^{\COM}_v = X(I-X)=0.
\]
Thus no COM-admissible valuation can make $\SEPone$ nonzero, so
\[
\neg\Sat_{\COM}^d(\SEPone).
\]
\end{proof}

\begin{lemma}[Definedness of \texorpdfstring{$\SEPone$}{SEP-1} in \(\PBA\) forces pairwise commutativity]\label{lem:sep1-definedness}
Let $v$ be a projector valuation such that $\sem{\SEPone}^{\PBA}_v$ is defined.
Then the three atom values
\[
A:=v(p),\qquad B:=v(q),\qquad C:=v(r)
\]
commute pairwise.
\end{lemma}

\begin{proof}
Because $q\vee r$ is a subformula of $\SEPone$, its PBA value is defined.
Hence $B\# C$.
Because $p\wedge q$ and $p\wedge r$ are subformulas of $\SEPone$, their PBA values are also defined.
Hence $A\# B$ and $A\# C$.
Therefore $A,B,C$ commute pairwise.
\end{proof}

\begin{theorem}[\(\SEPone\) is PBA-unsatisfiable]\label{thm:sep1-pba-unsat}
For every $d\geq 1$,
\[
\neg\Sat_{\PBA}^d(\SEPone).
\]
\end{theorem}

\begin{proof}
Suppose that $\sem{\SEPone}^{\PBA}_v$ is defined for some projector valuation $v$.
By Lemma~\ref{lem:sep1-definedness}, the atom values $v(p),v(q),v(r)$ commute pairwise, so $v$ is COM-admissible for $\SEPone$.
By Lemma~\ref{lem:com-to-pba-eval}, the PBA and COM values agree on $\SEPone$:
\[
\sem{\SEPone}^{\PBA}_v = \sem{\SEPone}^{\COM}_v.
\]
Theorem~\ref{thm:sep1-com-unsat} shows that the right-hand side is always $0$.
Hence every defined PBA evaluation of $\SEPone$ has value $0$.
Therefore
\[
\neg\Sat_{\PBA}^d(\SEPone).
\]
\end{proof}

\begin{theorem}[\(\SEPone\) is STD-satisfiable]\label{thm:sep1-std-sat}
For every $d\geq 2$,
\[
\Sat_{\STD}^d(\SEPone).
\]
\end{theorem}

\begin{proof}
Let $e_1,e_2,\dots,e_d$ be the standard basis of $H$.
Define subspaces
\[
P:=\operatorname{span}(e_1+e_2),
\qquad
Q:=\operatorname{span}(e_1),
\qquad
R:=\operatorname{span}(e_2).
\]
Let $v$ be the standard valuation satisfying
\[
 v(p)=P,
 \qquad
 v(q)=Q,
 \qquad
 v(r)=R,
\]
and sending all other atoms to $\{0\}$.
We compute the value of $\SEPone$ under $v$.

First,
\[
Q+R = \operatorname{span}(e_1,e_2),
\]
so $P\subseteq Q+R$ and therefore
\[
\sem{p\wedge(q\vee r)}^{\STD}_v = P\cap(Q+R)=P.
\]

Next we show that
\[
P\cap Q = \{0\},
\qquad
P\cap R = \{0\}.
\]
Let $x\in P\cap Q$.
Then for some scalars $a,b\in\field$,
\[
 x=a(e_1+e_2)=be_1.
\]
Comparing coordinates gives $a=0$ from the $e_2$-coordinate and then $b=0$ from the $e_1$-coordinate.
Thus $x=0$, so $P\cap Q=\{0\}$.
The proof of $P\cap R=\{0\}$ is identical: if
\[
 x=a(e_1+e_2)=ce_2,
\]
then comparison of coordinates gives $a=c=0$.

Hence
\[
\sem{(p\wedge q)\vee(p\wedge r)}^{\STD}_v
=
(P\cap Q)+(P\cap R)
=
\{0\}.
\]
Taking orthocomplement yields
\[
\sem{\neg((p\wedge q)\vee(p\wedge r))}^{\STD}_v = H.
\]
Therefore
\[
\sem{\SEPone}^{\STD}_v
=
\sem{p\wedge(q\vee r)}^{\STD}_v \cap \sem{\neg((p\wedge q)\vee(p\wedge r))}^{\STD}_v
=
P\cap H = P.
\]
Since $P\neq\{0\}$, we conclude that
\[
\sem{\SEPone}^{\STD}_v\neq\{0\}.
\]
Thus $\Sat_{\STD}^d(\SEPone)$.
\end{proof}

\section{Separation theorems}

\begin{theorem}[Explicit separation of \(\STD\) from \(\COM\)]\label{thm:std-not-implies-com}
For every $d\geq 2$ there exists a formula $\varphi$ such that
\[
\Sat_{\STD}^d(\varphi)
\qquad\text{and}\qquad
\neg\Sat_{\COM}^d(\varphi).
\]
More precisely, one may take $\varphi=\SEPone$.
\end{theorem}

\begin{proof}
Combine Theorems~\ref{thm:sep1-com-unsat} and~\ref{thm:sep1-std-sat}.
\end{proof}

\begin{theorem}[Explicit separation of \(\STD\) from \(\PBA\)]\label{thm:std-not-implies-pba}
For every $d\geq 2$ there exists a formula $\varphi$ such that
\[
\Sat_{\STD}^d(\varphi)
\qquad\text{and}\qquad
\neg\Sat_{\PBA}^d(\varphi).
\]
More precisely, one may take $\varphi=\SEPone$.
\end{theorem}

\begin{proof}
Combine Theorems~\ref{thm:sep1-pba-unsat} and~\ref{thm:sep1-std-sat}.
\end{proof}

\begin{corollary}[Comparison of satisfiability classes]\label{cor:comparison-classes}
For each fixed $d\geq 2$, let
\[
\mathsf{SAT}_{X}^d := \{\varphi : \Sat_X^d(\varphi)\}
\qquad (X\in\{\COM,\PBA,\STD\}).
\]
Then
\[
\mathsf{SAT}_{\COM}^d \subseteq \mathsf{SAT}_{\PBA}^d \subsetneq \mathsf{SAT}_{\STD}^d,
\qquad
\mathsf{SAT}_{\COM}^d \subsetneq \mathsf{SAT}_{\STD}^d.
\]
The only unresolved relation among these three classes is whether
\[
\mathsf{SAT}_{\COM}^d \subsetneq \mathsf{SAT}_{\PBA}^d
\]
or
\[
\mathsf{SAT}_{\COM}^d = \mathsf{SAT}_{\PBA}^d.
\]
\end{corollary}

\begin{proof}
The inclusions
\[
\mathsf{SAT}_{\COM}^d \subseteq \mathsf{SAT}_{\PBA}^d \subseteq \mathsf{SAT}_{\STD}^d
\]
follow from Corollary~\ref{cor:implication-chain}.
Strictness of
\[
\mathsf{SAT}_{\PBA}^d \subsetneq \mathsf{SAT}_{\STD}^d
\]
and
\[
\mathsf{SAT}_{\COM}^d \subsetneq \mathsf{SAT}_{\STD}^d
\]
follows from Theorems~\ref{thm:std-not-implies-pba} and~\ref{thm:std-not-implies-com}, respectively.
No strictness result between $\COM$ and $\PBA$ has been established in this paper.
\end{proof}

\section{Discussion and open problem}

The comparison picture obtained here is already nontrivial.
The standard Hilbert-lattice semantics is strictly more permissive than both the global commuting semantics and the local partial-Boolean semantics, and the witness formula is completely explicit.
The proof is short because the separator isolates the distributive collapse that appears inside any commuting Boolean block while exploiting the genuine non-distributivity of the subspace lattice in the standard semantics.

What remains open is the relation between local commeasurability and global commuting realizability.
The definitions make it plausible that the local partial-Boolean semantics should be strictly more permissive than the global commuting semantics, but the present paper does not prove that.
Any such result will require a formula whose nodewise commutativity constraints do not already collapse the valuation into one global commuting family.

\begin{problem}
Fix $d\geq 2$.
Determine whether
\[
\mathsf{SAT}_{\COM}^d = \mathsf{SAT}_{\PBA}^d
\]
or
\[
\mathsf{SAT}_{\COM}^d \subsetneq \mathsf{SAT}_{\PBA}^d.
\]
Equivalently, determine whether there exists a formula $\psi$ such that
\[
\Sat_{\PBA}^d(\psi)
\qquad\text{and}\qquad
\neg\Sat_{\COM}^d(\psi).
\]
\end{problem}

\paragraph{Positioning.}
This paper should be read as a semantics-comparison note.
It does not claim a new generic reduction of fixed-dimensional standard quantum-logic satisfiability to algebraic feasibility; that terrain is already substantially occupied by Herrmann's modular-lattice complexity results~\cite{Herrmann2022}.
It does not claim a new generic field-to-lattice translation theorem; translation and definability technology already exist in Herrmann--Ziegler~\cite{HerrmannZiegler2018,HerrmannZiegler2019}.
It also does not identify the partial-Boolean semantics with the standard Hilbert-lattice semantics; the recent complexity results of Dawar--Shah concern the former, not the latter~\cite{DawarShah2026}.
The contribution here is narrower and cleaner: it proves precise one-way implications and explicit strict separations among the three fixed-dimension satisfiability notions currently under comparison.

\end{document}